\documentclass[onecolumn,journal,doublespace]{IEEEtran}
\usepackage{indentfirst}
\usepackage{xcolor,color}
\usepackage{amsfonts,amsmath,amssymb,amsthm}
\usepackage{graphicx,multirow,bm}
\usepackage{cite}
\usepackage{arydshln}

\newtheorem{Theorem}{Theorem}
\newtheorem{Example}{Example}
\newtheorem{Remark}{Remark}
\newtheorem{Lemma}{Lemma}
\newtheorem{Construction}{Construction}
\newtheorem{Proposition}{Proposition}
\begin{document}
\title{New Minimum Storage Regenerating Codes With Small Sub-packetization Levels for All Admissible Repair Degrees
}

\author{Jie~Li,~\IEEEmembership{Senior~Member,~IEEE,} 
Xiaohu~Tang,~\IEEEmembership{Fellow,~IEEE,}
Linqi~Song,~\IEEEmembership{Senior~Member,~IEEE}
\thanks{J. Li and L. Song are with the City University of Hong Kong, Hong Kong SAR, China (e-mails: jieli873@gmail.com; 	linqi.song@cityu.edu.hk).}
\thanks{X. Tang is with the Information Coding and Transmission Key Lab of Sichuan Province, CSNMT Int. Coop. Res. Centre (MoST), Southwest Jiaotong University, Chengdu, 610031, China (e-mail: xhutang@swjtu.edu.cn).}
}

\maketitle

\begin{abstract}
In this paper, {we propose a new family of  $(n, k)$ Minimum Storage Regenerating (MSR) codes} for distributed storage systems. By structuring the parity-check matrix as a block Vandermonde matrix composed of diagonal matrices and generalized permutation matrices, the proposed code achieves a sub-packetization level of $N = (d-k+1)^{\lceil\frac{n}{2}\rceil}$ and supports any repair degree $k+1\le d<n$. A key innovation of the new construction lies in relaxing the strict commutativity condition previously required for the building blocks in the parity-check matrices of Ye-Barg codes. We introduce a novel proof technique to establish the MDS property without relying on full commutativity, thereby generalizing the applicability of block Vandermonde matrices and enriching the family of Ye-Barg codes. Additionally, the sparsity of the parity-check matrix enhances encoding and decoding efficiency by reducing computational complexity. This work provides new insights into the construction of MSR codes.
\end{abstract}


\section{Introduction}
In the rapidly evolving field of distributed storage systems, the need for efficient data management strategies has become increasingly critical. With vast amounts of data generated daily, ensuring reliability and performance has taken center stage. Distributed storage systems employing Maximum Distance Separable (MDS) codes achieve an optimal trade-off between fault tolerance and storage overhead, making them widely used over the past few decades. 

Specifically, 
an $(n,k)$ MDS code over a finite field $\mathbf{F}_q$ can be deployed in distributed storage systems as follows: A file consisting of $k$ symbols in $\mathbf{F}_q$ is encoded using an $(n,k)$ MDS code, resulting in a codeword of $n$ symbols, which are then distributed across $n$ storage nodes. This setup allows the original file to be recovered from any $k$ out of the $n$ nodes. However, when repairing a failed node, one must contact $k$ surviving nodes and download $k$ symbols from $\mathbf{F}_q$ (assuming no subsymbols from a subfield of $\mathbf{F}_q$ are considered). This leads to excessive \textit{repair bandwidth}, defined as the amount of data downloaded to regenerate a failed node.

In the seminal work of Dimakis et al. \cite{dimakis2010network}, Minimum Storage Regenerating (MSR) codes were introduced as a class of MDS array codes, where a codeword is represented as an array of size 
$N\times n$, with $N$ referred to as the \textit{sub-packetization level}.
MSR codes have emerged as a transformative approach to mitigating the excessive repair bandwidth associated with MDS codes, optimizing data communication during node repair while minimizing storage overhead.
In \cite{dimakis2010network}, it was proven that the repair bandwidth of $(n,k)$ MSR codes with sub-packetization level $N$ is lower bounded by
\begin{equation}\label{Eqn_gamma_opt}
\gamma\ge \gamma_{\rm optimal}=\frac{d}{d-k+1}N,
\end{equation}
where $d$ ($k\le d<n$) denotes the number of helper nodes contacted during the repair process, known as the \textit{repair degree}. 

MSR codes significantly reduce the repair bandwidth compared to conventional MDS codes. The construction of MSR codes has attracted considerable attention since 2007, leading to substantial progress over the past decade \cite{Zigzag,li2023MDS,liu2023generic,wang2012long,hadamard,Hadamard strategy,invariant subspace,Long IT,Goparaju,tian2014characterizing,Barg1,Barg2,Sasidharan-Kumar2,Rashmi2011optimal,han2015update,li2017generic,li2018generic,hou2020binary,elyasi2020cascade,li2022generic,liu2022generic,vajha2023small,wang2023constructing,zhang2023vertical,chen2020explicit,li2024msr}. 
In \cite{Rashmi2011optimal}, MSR codes for general parameters were constructed for low code rates. In \cite{Zigzag, wang2012long,hadamard, Hadamard strategy, invariant subspace, Long IT}, high-rate MSR codes were obtained, although they are explicitly applicable only to two or three parity nodes. Significant advances in high-rate MSR codes with general parameters were made only in 2016. In \cite{li2017generic,Sasidharan-Kumar2,li2018generic,Barg2}, $(n ,k)$ MSR codes over non-binary fields for any $n$ and $k$ were independently proposed with a sub-packetization level of $(n-k)^{\frac{n}{n-k}}$ and $d=n-1$. Subsequently, 
$(n ,k)$ MSR codes over binary fields for any $n$ and $k$ were also independently introduced in \cite{hou2020binary} and  \cite{li2022generic}.

Later, for $(n, k)$ MSR codes with repair degree $d<n-1$, in their seminal work \cite{Barg1}, Ye and Barg introduced two families of $(n, k)$ MSR codes by constructing the parity-check matrix as a block Vandermonde matrix, using diagonal matrices and generalized permutation matrices as building blocks, respectively. In \cite{chen2020explicit}, Chen and Barg presented an $(n, k)$ MSR code with a sub-packetization level of $(d-k+1)^n$ using upper triangular matrices as the components of the parity-check matrix. In \cite{liu2022generic}, a generic transformation to construct $(n, k)$ MDS codes with repair degree $d\le n-1$ was proposed, achieving a sub-packetization level of $(d-k+1)^{\lceil \frac{n}{2}\rceil}$. In \cite{vajha2023small}, optimal-access MSR codes were proposed, which only support repair degrees of $d=k+1, k+2, k+3$. Subsequently, in independent works \cite{zhang2023vertical} and \cite{li2023MDS}, $(n, k)$ MSR codes with a sub-packetization level of $(d-k+1)^{\lceil\frac{n}{2}\rceil}$ were proposed using combinations of upper triangular and diagonal matrices. 
In \cite{wang2023constructing}, $(n, k)$ MSR codes with a small sub-packetization level of $2^{n/3}$ were proposed, but the repair degree was fixed to $d=k+1$. This was later generalized in \cite{li2024msr} to support any repair degree $d<n-1$, albeit requiring relatively large finite fields for general code parameters. 

Notably, the Ye–Barg codes constitute two broad families of MSR codes, where their parity-check matrices are composed of commuting building blocks. The MDS property of the Ye–Barg code in \cite[Section VI]{Barg1} relies on this commutativity property together with the block Vandermonde structure of the parity-check matrix.
The Ye–Barg codes and the block Vandermonde structure with commuting building blocks have inspired numerous subsequent constructions, such as those in \cite{chen2020explicit, eMSR_d_eq_n-1, ye2020new, liu2023new, li2024constructions, wang2023rack}.
For example, the parity-check matrices of the binary MDS array codes in \cite{li2024constructions} also exhibit the block Vandermonde structure and are composed of commuting building blocks similar to those of the Ye–Barg codes, and the parity-check matrix of the rack-aware MSR codes in \cite{wang2023rack} shares a similar structure with the Ye–Barg code in \cite[Section VI]{Barg1}.

In this paper, we propose a new MSR code whose parity-check matrix is also structured as a block Vandermonde matrix. Specifically, the parity-check matrix is formulated by combining diagonal matrices and generalized permutation matrices, yielding an $(n, k)$ MSR code with a sub-packetization level of $N=(d-k+1)^{\lceil\frac{n}{2}\rceil}$. The proposed code can be regarded as a concatenated version of the Ye–Barg codes introduced in \cite{Barg1}, yet it achieves a smaller sub-packetization level than the original Ye–Barg codes, while requiring a finite field size $q$ satisfying $q>\lceil\frac{n}{2}\rceil w$ with $\gcd(w, q-1)=1$.

Unlike the Ye–Barg constructions, the building blocks of the parity-check matrix in the new MSR code are not always commutative. This is the first time the strict commutativity condition for block Vandermonde matrices, as required in \cite[Lemma 14]{Barg1}, has been relaxed while still preserving the MDS property. We introduce a novel proof technique, previously unexplored in the literature, to establish the MDS property of the proposed code without relying on the strict commutativity condition. This result suggests that Lemma 14 in \cite{Barg1} may be generalized to guarantee the nonsingularity of block Vandermonde matrices even in certain non-commutative settings.

The proposed MSR code broadens the Ye–Barg code families and offers new insights into MSR code design. It could potentially serve as a building block for constructing new MDS array codes with small repair bandwidth, rack-aware MSR codes, and binary MSR codes, similar to those in \cite{chen2020explicit, eMSR_d_eq_n-1, ye2020new, liu2023new, li2024constructions, wang2023rack} that employ the Ye–Barg codes as building blocks. Furthermore, the parity-check matrix of the new code is sparse, which substantially enhances encoding and decoding efficiency by reducing the number of finite field multiplications required for matrix-vector computations.

The remainder of this paper is organized as follows. Section \ref{sec:pre} reviews some necessary preliminaries on MDS array codes. Section \ref{sec:construction} presents the construction of the MSR code with a small sub-packetization level and states the main results. The MDS property and the repair property are established in Sections \ref{sec:MDS} and \ref{sec:repair}, respectively. Section \ref{sec:comp} compares key parameters between the proposed MSR code and existing MSR codes, and Section \ref{sec:conclusion} concludes the paper.

\section{Preliminaries}\label{sec:pre}

In this section, we introduce some preliminaries on MDS array codes and a special partition of a given basis. Let $q$ be a prime power and $\mathbf{F}_q$ denote the finite field with $q$ elements. For two integers $a$ and $b$ with $a<b$, let $[a: b)$ represent the set $\{a, a+1, \ldots, b-1\}$. 

\subsection{$(n,k)$ Array Codes}

An $(n,k)$ array code can be defined by an $rN\times n N$ matrix $A$ of full rank over $\mathbf{F}_q$, where the codeword 
$$\begin{pmatrix}
\mathbf{f}_0,
\mathbf{f}_1,
\ldots,
\mathbf{f}_{n-1}   
\end{pmatrix}$$ satisfies
\begin{equation}\label{Eqn parity check eq}
A
\left(\begin{array}{c}
\mathbf{f}_0\\
\mathbf{f}_1\\
\vdots\\
\mathbf{f}_{n-1}
\end{array}
\right)=\mathbf{0}_{rN},
\end{equation}
with $r=n-k$, $\mathbf{f}_i$ is a column vector of length $N$ over $\mathbf{F}_q$, and $\mathbf{0}_{rN}$ denotes the zero column vector of length $rN$ (abbreviated as $\mathbf{0}$ when its length is clear from context). The $rN\times nN$ block matrix $A$ in \eqref{Eqn parity check eq} is called the \textit{parity-check matrix} of the code and can be represented as a block matrix
\begin{align*}
 A=&(A_{t,i})_{t\in [0: r),i\in[0: n)}
 \\=&\left(
\begin{array}{cccc}
A_{0,0} & A_{0,1} & \cdots & A_{0,n-1} \\
A_{1,0} & A_{1,1} & \cdots & A_{1,n-1} \\
\vdots & \vdots & \ddots & \vdots \\
A_{r-1,0} & A_{r-1,1} & \cdots & A_{r-1,n-1}
\end{array}\right),
\end{align*}
where each $A_{i,j}$ is a square matrix of order $N$. 
For each $t\in[0: r)$, the equation $\sum\limits_{i=0}^{n-1}A_{t,i}\mathbf{f}_i=0$ in \eqref{Eqn parity check eq} yields $N$ equations. We refer to $\sum\limits_{i=0}^{n-1}A_{t,i}\mathbf{f}_i=0$ as the $t$-th \textit{parity-check group}.

In this paper, we always assume 
\begin{equation}\label{Eqn_DRF}
A_{0,j}=I_N \quad \text{for } j\in [0: n)
\end{equation}
so that the MDS array codes possess the degraded-read-friendly property \cite{wu2021achievable}, where $I_N$ denotes the identity matrix of order $N$ (abbreviated as $I$ when the order is clear from context).

Previously, many constructions relied on a block Vandermonde structure of the parity-check matrix while requiring that all block matrices pairwise commute to preserve the MDS property \cite{eMSR_d_eq_n-1, li2024constructions, wang2023rack}. In this paper, we propose a new MSR code that relaxes this strict commutativity condition while still preserving the MDS property, as will be shown in the following sections.

Additionally, we introduce the concept of \textit{density}, defined as the fraction of nonzero entries in a parity-check matrix. It accounts for the number of finite field multiplications required for matrix-vector computations during encoding and decoding, thus serving as a measure of encoding/decoding efficiency. Formally, the density of an $(n,k)$ array code with sub-packetization level $N$ is defined as
\begin{equation*}
D = \frac{N_{\rm nonzero}}{r n N^2},
\end{equation*}
where $N_{\rm nonzero}$ denotes the number of nonzero elements in the parity-check matrix. Codes with small density are preferred.

\subsection{Repair Mechanism}\label{sec:repair_mech}

A repair mechanism for $(n,k)$ array codes with repair degree $d=n-1$ was formulated in \cite{eMSR_d_eq_n-1}. We revisit it in this subsection.

For an $(n,k)$ array code, suppose node $i$ ($i\in[0: n)$) fails. Similar to \cite{eMSR_d_eq_n-1}, we consider the following linear equations based on \eqref{Eqn parity check eq}:

\begin{equation}\label{Eqn repairment gene}
\underbrace{\left(\begin{array}{c}
S_{i} A_{0,i} \\
S_{i} A_{1,i}\\
\vdots\\
S_{i} A_{r-1,i}
\end{array}\right)\mathbf{f}_i}_{\mathrm{useful ~data}}+\sum\limits_{j=0,j\ne i}^{n-1}\underbrace{\left(\begin{array}{c}
S_{i} A_{0,j}\\
S_{i} A_{1,j}\\
\vdots\\
S_{i} A_{r-1,j}
\end{array}\right)\mathbf{f}_j}_{\mathrm{interference ~by~}\mathbf{f}_j}=0,
\end{equation}
where $S_{i}$ is an $\frac{N}{r}\times N$ full-rank matrix, referred to as the \textit{selection matrix}. 

According to \eqref{Eqn repairment gene}, node $i$ achieves optimal repair bandwidth if the following two requirements are satisfied:
\begin{itemize}
\item [R1.] The coefficient matrix of the system of linear equations (with respect to the variable $\mathbf{f}_i$) in \eqref{Eqn repairment gene} has rank $N$, i.e.,
\begin{equation}\label{repair_node_requirement1 d}
\operatorname{rank}\left(\begin{array}{c}
S_{i} A_{0,i}\\
S_{i} A_{1,i}\\
\vdots\\
S_{i} A_{r-1,i}
\end{array}\right)=N.
\end{equation}

\item [R2.] The interference caused by $\mathbf{f}_j$ can be resolved by downloading $N/r$ symbols (i.e., $S_{i}\mathbf{f}_j$) from node $j$ for $j\in [0: n)\backslash\{i\}$, i.e.,
\begin{equation*}
\operatorname{rank}\left(\begin{array}{c}
S_{i}A_{0,j}\\
S_{i}A_{1,j}\\
\vdots\\
S_{i}A_{r-1,j}
\end{array}\right)=N/r, \quad j\in [0: n)\backslash\{i\},
\end{equation*}
or equivalently,
\begin{equation*}
\operatorname{rank}\left(\begin{array}{c}
S_{i}\\
S_{i}A_{t,j}
\end{array}\right)=N/r, \quad t\in [1: r), \; j\in [0: n)\backslash\{i\},
\end{equation*}
by \eqref{Eqn_DRF}.
\end{itemize}

For a detailed characterization of repair mechanisms for $(n,k)$ array codes with repair degree $d<n-1$, interested readers may refer to \cite{li2023MDS}. Interestingly, in this work, the repair mechanism for the proposed $(n,k)$ array codes with repair degree $d<n-1$ can be reduced to that of a code with repair degree $d=n-1$.

{
\begin{Remark} In \cite{eMSR_d_eq_n-1}, the repair mechanism uses a double-subscript selection matrix 
$S_{i,t}$
  applied to the parity-check groups, where $i$ is the index of the failed node and $t$ is the parity-check group index. In our construction we use a single-subscript $S_i$, which selects parity-check equations from each pariy-check group in the same way for a given failed node $i$. This is not restrictive for the codes considered here. Since the extra freedom of $S_{i,j}$
  is unnecessary for the proposed code to achieve optimal repair bandwidth.
\end{Remark}

}

\subsection{Partition of the Basis $\{e_0,\dots,e_{w^m-1}\}$}
Let $w$ and $m$ be integers greater than or equal to $2$. Denote by $e_0,\dots,e_{w^m-1}$ a basis of $\mathbf{F}_q^{w^m}$.
For simplicity, we regard them as the standard basis, i.e.,
\begin{align*}
e_i=(0,\dots,0,1,0,\dots,0), \quad i\in [0: w^m),
\end{align*}
where only the $i$-th entry is nonzero. Throughout this paper,
let $N=w^m$ unless otherwise stated. 
In the following, we revisit a class of special partition sets of $\{e_0,\dots,e_{w^m-1}\}$ originally introduced in \cite{wang2012long} and later reformulated in \cite{eMSR_d_eq_n-1}, which plays an important role in our proposed construction.

Given an integer $a\in [0: w^m)$, denote its $w$-ary expansion by $(a_0,\dots,a_{m-1})$
with $a_0$ being the most significant digit, i.e., $a=\sum_{j=0}^{m-1}w^{m-1-j}a_{j}$. For convenience, we also write $a=(a_0,\dots,a_{m-1})$.
For $i\in [0: m)$ and $t\in [0: w)$, define a subset of $\{e_0,\dots,e_{w^m-1}\}$ as
\begin{equation}\label{Eqn_Vt}
V_{i,t}=\{e_a\mid a_i=t,\; 0\le a< w^m\},
\end{equation}
where $a_i$ is the $i$-th element in the $w$-ary expansion of $a$.
Clearly, $V_{i,0},V_{i,1},\dots,V_{i,w-1}$ form a partition of $\{e_0,\dots,e_{w^m-1}\}$ and $|V_{i,t}|=w^{m-1}$ for any $i\in [0 : m)$. Table \ref{example partition} gives two examples of the set partitions defined in \eqref{Eqn_Vt}.
\begin{table}[htbp]
\begin{center}
\caption{(a) and (b) denote the $m$  partitions of  the set $\{e_0,\cdots,e_{w^m-1}\}$    defined by \eqref{Eqn_Vt} for $m=3,w=2$, and $m=2,w=3$, respectively.}
\label{example partition}\begin{tabular}{|c|c|c|c|c|c|c|c|}
\hline $i$ & 0 & 1 & 2 & $i$ & 0 & 1 & 2\\
\hline \multirow{4}{*}{$V_{i,0}$ }&$e_0$&$e_0$&$e_0$&\multirow{4}{*}{$V_{i,1}$ }&$e_4$&$e_2$&$e_1$\\
  & $e_1$&$e_1$&$e_2$ && $e_5$&$e_3$&$e_3$\\
  &$e_2$&$e_4$&$e_4$&& $e_6$&$e_6$&$e_5$\\
  &$e_3$&$e_5$&$e_6$&& $e_7$&$e_7$&$e_7$\\
\hline\multicolumn{8}{c}{\hspace{4mm}(A)}
\end{tabular}\\\vspace{5mm}
\begin{tabular}{|c|c|c|c|c|c|c|c|c|c}
\hline $i$ & 0 & 1 & $i$ & 0 & 1 & $i$ & 0 & 1\\
\hline \multirow{3}{*}{$V_{i,0}$ }&$e_0$&$e_0$&\multirow{3}{*}{$V_{i,1}$ }&$e_3$&$e_1$&\multirow{3}{*}{$V_{i,2}$ }&$e_6$&$e_2$\\
  & $e_1$&$e_3$&& $e_4$&$e_4$&& $e_7$&$e_5$\\
  & $e_2$&$e_6$&& $e_5$&$e_7$&& $e_8$&$e_8$\\
\hline\multicolumn{6}{c}{\hspace{25mm}(B)}
\end{tabular}
\end{center}
\end{table}

For convenience of notation, we further define
\begin{equation}\label{Eqn_Vt_more}
V_{i+sm,t}=V_{i,t} \quad \text{for } i\in[0: m),\; t\in[0: w),\; s\ge1.
\end{equation}

Furthermore, for any $0\le i_1, i_2< sm$ and $i_1\not \equiv i_2\mbox{\ mod\ } m$, we define $V_{i_1,i_2,t_1,t_2}=V_{i_2,i_1,t_2,t_1}=V_{i_1,t_1}\cap V_{i_2,t_2}$, i.e.,
\begin{eqnarray*}\label{Eqn_Vit}
V_{i_1,i_2,t_1,t_2}= V_{i_2,i_1,t_2,t_1}=\{e_a|a_{i_1}=t_1, ~a_{i_2}=t_2, ~a\in [0,~ w^m)\},
\end{eqnarray*}
where $0\le t_1,t_2<w$.
Then, we have
\begin{eqnarray}\label{Eqn_B3}
V_{i_1,t_1}&=&V_{i_1,i_2,t_1,0}\cup\cdots \cup V_{i_1,i_2,t_1,w-1}.
\end{eqnarray}
For convenience of notation,  we also denote by $V_{i_1,t_1}$ and $V_{i_1,i_2,t_1,t_2}$  the $w^{m-1}\times w^m$ and {$w^{m-2}\times w^m$} matrices,
 whose rows are formed by vectors $e_i$
in their corresponding sets, respectively,  such that $i$  is sorted in ascending order.  For example, when $w=2$ and $m=3$, $V_{1,0}$ can be viewed as a $4\times 8$ matrix as follows
\begin{eqnarray*}
V_{1,0}=\left(e_0^{\top}~ e_1^{\top} ~e_4^{\top}~ e_5^{\top}\right)^{\top},
\end{eqnarray*}
where $\top$ represents the transpose operator.

\section{New Construction of an MSR Code With a Small Sub-packetization Level}\label{sec:construction}

\begin{Construction}\label{Con1}
For positive integers $n,k,m$ such that $n=2m$\footnote{An $(n, k)$ array code with $2\nmid n$ can be obtained by shortening an $(n+1, k+1)$ array code. Therefore, we assume $2\mid n$ for notational convenience.} and $k<n-1$, let $k<d\le n-1$, $r=n-k$, $w=d-k+1$ and $N=w^{m}$. The parity-check matrix $(A_{t,i})_{t\in [0: r),i\in[0: n)}$ of the new $(n,k)$ array code $\mathcal{C}$ over a finite field  $\mathbf{F}_q$ is defined by
\begin{equation}\label{Eqn_A_ti}
    A_{t,i}=A_i^t, \quad t\in [0: r),\; i\in [0:n)
\end{equation}
and
\begin{equation}\label{Eqn_matrix_A}
\left(\begin{array}{c}
     V_{i,0}  \\
     V_{i,1}\\
     \vdots\\
     V_{i,w-1}
\end{array}\right)A_i=
\left\{
\begin{array}{ll}
\left(\begin{array}{c}
    \lambda_{i,0} V_{i,0}  \\
   \lambda_{i,1}  V_{i,1}\\
     \vdots\\
   \lambda_{i,w-1}  V_{i,w-1}
\end{array}\right), &\text{if } i\in [0: m), \\[30pt]
\left(\begin{array}{c}
    \lambda_{i} V_{i,1}  \\
     V_{i,2}\\
     \vdots\\
     V_{i,w-1}\\
    V_{i,0}
\end{array}\right)   , & \text{if } i\in [m: n),\\
\end{array}
\right.
\end{equation}
where $\lambda_{i,u}$, $i\in [0:m)$, $u\in [0: w)$, and $\lambda_j$, $j\in [m: n)$ are nonzero elements in $\mathbf{F}_q$. 
\end{Construction}

{
Together with Table \ref{example partition}, we provide a toy example of the structure of the new code $\mathcal{C}$ as follows.
\begin{Example}
Let $m=3, n=6, k=3$, and $d=4$, then the parity-check matrix $(A_{t,i})_{t\in [0: 3),i\in[0: 6)}$ of the new $(6, 3)$ array code $\mathcal{C}$ in Construction \ref{Con1} are specified as follows:
\begin{align*}
A_0=&{\rm diag}(\lambda_{0,0},\lambda_{0,0},\lambda_{0,0},\lambda_{0,0},\lambda_{0,1},\lambda_{0,1},\lambda_{0,1},\lambda_{0,1}), \quad A_1={\rm diag}(\lambda_{1,0},\lambda_{1,0},\lambda_{1,1},\lambda_{1,1},\lambda_{1,0},\lambda_{1,0},\lambda_{1,1},\lambda_{1,1}),  \\   
A_2=&{\rm diag}(\lambda_{2,0},\lambda_{2,1},\lambda_{2,0},\lambda_{2,1},\lambda_{2,0},\lambda_{2,1},\lambda_{2,0},\lambda_{2,1}),    
\end{align*}
\begin{align*}
A_3=&\begin{pmatrix}
&&&&\lambda_3&&&\\
&&&&&\lambda_3&&\\
&&&&&&\lambda_3&\\
&&&&&&&\lambda_3\\
1&&&&&&&\\
&1&&&&&&\\
&&1&&&&&\\
&&&1&&&&\\
\end{pmatrix}, A_4=\begin{pmatrix}
&&\lambda_4&&&&&\\
&&&\lambda_4&&&&\\
1&&&&&&&\\
&1&&&&&&\\
&&&&&&\lambda_4&\\
&&&&&&&\lambda_4\\
&&&&1&&&\\
&&&&&1&&\\
\end{pmatrix},\\
A_5=&\begin{pmatrix}
&\lambda_5&&&&&&\\
1&&&&&&&\\
&&&\lambda_5&&&&\\
&&1&&&&&\\
&&&&&\lambda_5&&\\
&&&&1&&&\\
&&&&&&&\lambda_5\\
&&&&&&1&\\
\end{pmatrix}.     
\end{align*}
\end{Example}
}

\begin{Theorem}\label{Thm_MDS}
The $(n,k)$ array code $\mathcal{C}$ defined in Construction \ref{Con1} is MDS if
\begin{itemize}
\item [i)] $\lambda_{i,u}\ne \lambda_{j,v}$ for $i,j\in [0:m)$ with $i\ne j$, and $u,v\in [0: w)$;
\item  [ii)] $\lambda_{i,t}^w\ne \lambda_j$ for $i\in [0:m)$, $t\in [0: w)$, and $j\in [m: n)\backslash \{i+m\}$;
\item [iii)] $\prod_{t=0}^{w-1}\lambda_{i,t}\ne \lambda_{i+m}$ for $i\in [0: m)$;
\item [iv)] $\lambda_{i}\ne \lambda_{j}$ for $i, j\in [m: n)$ with $i\ne j$.
\end{itemize}
\end{Theorem}

\begin{proof}
The proof is given in Section \ref{sec:MDS}.
\end{proof}

\begin{Theorem}\label{Thm_opt_RB}
The $(n,k)$ array code $\mathcal{C}$ defined in Construction \ref{Con1} achieves optimal repair bandwidth if $\lambda_{i,0}, \lambda_{i,1},\dots, \lambda_{i,w-1}$ are pairwise distinct and { there exists $\mu_i \in \mathbf{F}_q$ such that $\lambda_{i+m}=\mu_i^w$ for $i\in [0: m)$.}
\end{Theorem}
\begin{proof}
The proof is given in Section \ref{sec:repair}.
\end{proof}

By Theorems \ref{Thm_MDS} and \ref{Thm_opt_RB}, we have the following result.
\begin{Theorem}
The $(n, k)$ array code $\mathcal{C}$ defined in Construction \ref{Con1} is an MSR code over $\mathbf{F}_q$ if $q>\frac{wn}{2}$ with $\gcd(w, q-1)=1$, and by setting
\begin{equation}\label{Eq_lam_eq_n-1}
\lambda_{i,t}=c^{iw+t}, \quad i\in [0: m),\; t\in [0: w),
\end{equation}
and $\lambda_{i+m}=c^{iw^2}$ for $i\in [0:m)$.
\end{Theorem}

\begin{proof}
It suffices to show that the conditions in Theorems \ref{Thm_opt_RB} and \ref{Thm_MDS} are satisfied. From \eqref{Eq_lam_eq_n-1}, it is clear that $\lambda_{i,0}, \lambda_{i,1},\dots, \lambda_{i,w-1}$ are pairwise distinct and {$\lambda_{i+m}=(c^{iw})^w\in \mathbf{F}_q$} for any $i\in [0: m)$.

Now we verify conditions i)–iv) of Theorem \ref{Thm_MDS}. Condition i) is directly satisfied by \eqref{Eq_lam_eq_n-1}.
Let $\beta=c^w$, since $\gcd(w, q-1)=1$, $\beta$ is also a primitive element in $\mathbf{F}_q$.

For $i,j\in [0: m)$ with $i\ne j$ and $t\in [0: w)$,
we have $\lambda_{i,t}^w-\lambda_{j+m}=c^{iw^2+tw}-c^{jw^2}=c^{jw^2}(\beta^{(i-j)w+t}-1)\ne 0$,
since
\begin{equation*}
0<|(i-j)w+t|\le w(m-1)+w-1<q-1,
\end{equation*}
which verifies condition ii).

For $i\in [0: m)$,
we have $$\prod_{t=0}^{w-1}\lambda_{i,t}-\lambda_{i+m}=c^{iw^2+\frac{(w-1)w}{2}}-c^{iw^2}\ne 0$$ because $c^{\frac{(w-1)w}{2}}\ne 1$ given $0<\frac{(w-1)w}{2}<(w-1)m<q-1$, which confirms condition iii). Condition iv) is obviously satisfied. This completes the proof.
\end{proof}

\section{The MDS Property of the Array Code in Construction \ref{Con1}}\label{sec:MDS}

Before establishing the MDS property of the proposed array code, we present the following lemmas, which are instrumental in the verification.

\begin{Lemma}\label{Lemma block matrix}
(Block Vandermonde Matrix) Let $B_0, B_1, \dots, B_{r-1}$ be $N\times N$ matrices such that for any given $i\in [0: r)$, $B_i$ commutes with at least $r-2$ of the remaining $r-1$ matrices. Then the matrix
\begin{equation}\label{Eqn_blockVan}
  \left(
    \begin{array}{cccc}
      I & I & \cdots & I \\
      B_0 & B_1 & \cdots & B_{r-1} \\
      \vdots & \vdots & \ddots & \vdots \\
     B_0^{r-1} & B_1^{r-1} & \cdots & B_{r-1}^{r-1} \\
    \end{array}
  \right)
\end{equation}
is nonsingular if $B_i-B_j$ is nonsingular for all $i,j\in [0: r)$ with $i\ne j$.
\end{Lemma}

\begin{proof}
We first prove the claim that the matrix in \eqref{Eqn_blockVan} is nonsingular if
\begin{equation*}
\begin{pmatrix}
I&I&\cdots&I\\
B_{0} &B_{1}&\cdots&B_{{r-2}}\\ 
\vdots &\vdots&\ddots&\vdots\\    
B_{0}^{r-2} &B_{1}^{r-2}&\cdots&B_{{r-2}}^{r-2}\\    
\end{pmatrix}    
\end{equation*}
and {$B_{t}-B_{{r-1}}$ is nonsingular for $t\in [0: r-1)$}.

For a square matrix $A$, define $\operatorname{Inv}(A)=1$ if $A$ is nonsingular and $\operatorname{Inv}(A)=0$ otherwise.
Then we have
\begin{align*}
&\operatorname{Inv}\begin{pmatrix}
I&I&\cdots&I\\
B_{0} &B_{1}&\cdots&B_{{r-1}}\\ 
\vdots &\vdots&\ddots&\vdots\\    
B_{0}^{r-1} &B_{1}^{r-1}&\cdots&B_{{r-1}}^{r-1}\\    
\end{pmatrix}\\
=&
\operatorname{Inv}\begin{pmatrix}
I&I&\cdots&I&I\\
B_{0}-B_{{r-1}} &B_{1}-B_{{r-1}}&\cdots&B_{{r-1}}-B_{{r-1}}&\mathbf{0}\\ 
\vdots &\vdots&\ddots&\vdots&\vdots\\    
(B_{0}-B_{{r-1}})B_{0}^{r-2} &(B_{1}-B_{{r-1}})B_{1}^{r-2}&\cdots&{(B_{{r-2}}-B_{{r-1}})B_{r-2}^{r-2}}&\mathbf{0}\\    
\end{pmatrix}\\
= &
\operatorname{Inv}\begin{pmatrix}
B_{0}-B_{{r-1}} &B_{1}-B_{{r-1}}&\cdots&B_{{r-2}}-B_{{r-1}}\\ 
\vdots &\vdots&\ddots&\vdots\\    
(B_{0}-B_{{r-1}})B_{0}^{r-2} &(B_{1}-B_{{r-1}})B_{1}^{r-2}&\cdots&(B_{{r-2}}-B_{{r-1}})B_{{r-2}}^{r-2}\\    
\end{pmatrix}.
\end{align*}    
If $B_{r-1}$ commutes with all the remaining matrices $B_0,\dots,B_{r-2}$, then we have
\begin{align*}
  &  \operatorname{Inv}\begin{pmatrix}
B_{0}-B_{{r-1}} &B_{1}-B_{{r-1}}&\cdots&B_{{r-2}}-B_{{r-1}}\\ 
\vdots &\vdots&\ddots&\vdots\\    
(B_{0}-B_{{r-1}})B_{0}^{r-2} &(B_{1}-B_{{r-1}})B_{1}^{r-2}&\cdots&(B_{{r-2}}-B_{{r-1}})B_{{r-2}}^{r-2}    
\end{pmatrix}\\= &  \operatorname{Inv}\begin{pmatrix}
B_{0}-B_{{r-1}} &B_{1}-B_{{r-1}}&\cdots&B_{{r-2}}-B_{{r-1}}\\ 
\vdots &\vdots&\ddots&\vdots\\    
B_{0}^{r-2}(B_{0}-B_{{r-1}}) &B_{1}^{r-2}(B_{1}-B_{{r-1}})&\cdots&B_{{r-2}}^{r-2}(B_{{r-2}}-B_{{r-1}})    
\end{pmatrix}\\= & \left(\prod_{t=0}^{r-2}\operatorname{Inv}(B_{t}-B_{{r-1}})\right) \times  \operatorname{Inv}\begin{pmatrix}
I&I&\cdots&I\\
B_{0} &B_{1}&\cdots&B_{{r-2}}\\ 
\vdots &\vdots&\ddots&\vdots\\    
B_{0}^{r-2} &B_{1}^{r-2}&\cdots&B_{{r-2}}^{r-2}\\  
\end{pmatrix};
\end{align*}
Otherwise, $B_{r-1}$ does not commute with one of the remaining matrices, say $B_0$ (without loss of generality). Then we have
\begin{equation*}
B_{{r-1}}B_{t}=B_{t}B_{{r-1}}\quad\text{and}\quad (B_{0}-B_{{r-1}})^{-1}B_{t}=B_{t}(B_{0}-B_{{r-1}})^{-1}
\end{equation*}
for all $t\in [1: r-1)$ if $B_{0}-B_{{r-1}}$ is nonsingular. Thus,
\begin{align*}
  &  \operatorname{Inv}\begin{pmatrix}
B_{0}-B_{{r-1}} &B_{1}-B_{{r-1}}&\cdots&B_{{r-2}}-B_{{r-1}}\\ 
\vdots &\vdots&\ddots&\vdots\\    
(B_{0}-B_{{r-1}})B_{0}^{r-2} &(B_{1}-B_{{r-1}})B_{1}^{r-2}&\cdots&(B_{{r-2}}-B_{{r-1}})B_{{r-2}}^{r-2}    
\end{pmatrix}\\=&  \operatorname{Inv}\begin{pmatrix}
B_{0}-B_{{r-1}} &B_{1}-B_{{r-1}}&\cdots&B_{{r-2}}-B_{{r-1}}\\ 
\vdots &\vdots&\ddots&\vdots\\    
(B_{0}-B_{{r-1}})B_{0}^{r-2} &B_{1}^{r-2}(B_{1}-B_{{r-1}})&\cdots&B_{{r-2}}^{r-2}(B_{{r-2}}-B_{{r-1}})    
\end{pmatrix}\\= &\left(\prod_{t=1}^{r-2}\operatorname{Inv}(B_{t}-B_{{r-1}})\right) \times  \operatorname{Inv}\begin{pmatrix}
B_{0}-B_{{r-1}}&I&\cdots&I\\
(B_{0}-B_{{r-1}})B_{0} &B_{1}&\cdots&B_{{r-2}}\\ 
\vdots &\vdots&\ddots&\vdots\\    
(B_{0}-B_{{r-1}})B_{0}^{r-2} &B_{1}^{r-2}&\cdots&B_{{r-2}}^{r-2}
\end{pmatrix}\\
=&\left(\prod_{t=1}^{r-2}\operatorname{Inv}(B_{t}-B_{{r-1}})\right) \times   \operatorname{Inv}\begin{pmatrix}
I&(B_{0}-B_{{r-1}})^{-1}&\cdots&(B_{0}-B_{{r-1}})^{-1}\\
B_{0} &(B_{0}-B_{{r-1}})^{-1}B_{1}&\cdots&(B_{0}-B_{{r-1}})^{-1}B_{{r-2}}\\ 
\vdots &\vdots&\ddots&\vdots\\    
B_{0}^{r-2} &(B_{0}-B_{{r-1}})^{-1}B_{1}^{r-2}&\cdots&(B_{0}-B_{{r-1}})^{-1}B_{{r-2}}^{r-2}  
\end{pmatrix}\\
=&\left(\prod_{t=0}^{r-2}\operatorname{Inv}(B_{t}-B_{{r-1}})\right) \times   \operatorname{Inv}\begin{pmatrix}
I&I&\cdots&I\\
B_{0} &B_{1}&\cdots&B_{{r-2}}\\ 
\vdots &\vdots&\ddots&\vdots\\    
B_{0}^{r-2} &B_{1}^{r-2}&\cdots&B_{{r-2}}^{r-2}\\  
\end{pmatrix}.
\end{align*}
This confirms the claim. By induction, the matrix in \eqref{Eqn_blockVan} is nonsingular if {$B_{t}-B_{t'}$ is nonsingular for $0\le t<t'<r$}. This completes the proof.    
\end{proof}

\begin{Remark}
Lemma \ref{Lemma block matrix} relaxes the condition of Lemma 14 in \cite{Barg1}, which requires $B_0, B_1, \dots, B_{r-1}$ to be pairwise commutative.
\end{Remark}

\begin{Lemma}\label{lem-commute-matrix}
For the matrices $A_i$ defined in \eqref{Eqn_matrix_A}, we have 
\begin{itemize}
\item [i)] $A_iA_j=A_jA_i$ for $i,j\in [0:m)$ or $i,j\in [m:n)$ with $i\ne j$.
\item [ii)] $A_iA_j=A_jA_i$ for $i\in [0:m)$ and $j\in [m:n)$ such that $j\ne i+m$.
\end{itemize}
\end{Lemma}
\begin{proof}
i) The statement follows from the proof of Theorems 15 and 16 in \cite{Barg1} for $i,j\in [m:n)$.

 For $i,j\in [0:m)$ with $i\ne j$, and for any $u,v\in [0: w)$, we have
\begin{equation*}
V_{i,j,u,v}A_iA_j=\lambda_{i,u}V_{i,j,u,v}A_j=\lambda_{i,u}\lambda_{j,v}V_{i,j,u,v}= \lambda_{j,v} V_{i,j,u,v} A_i=V_{i,j,u,v}A_jA_i,
\end{equation*}
which implies $A_iA_j=A_jA_i$.

ii) 
{
For $i\in [0:m)$ and $j\in [m:n)$ such that $j\ne i+m$, we have
\begin{equation*}
\left(\begin{array}{c}
     V_{i,j,0,0}  \\
     V_{i,j,1,0}\\
     \vdots\\
     V_{i,j,w-1,0}
\end{array}\right)A_iA_j=
 \left(\begin{array}{c}
    \lambda_{i,0}V_{i,j,0,0}  \\
  \lambda_{i,1}   V_{i,j,1,0}\\
     \vdots\\
  \lambda_{i,w-1}   V_{i,j,w-1,0}
\end{array}\right)A_j=
\lambda_j \left(\begin{array}{c}
   \lambda_{i,0}V_{i,j,0,1}  \\
  \lambda_{i,1}   V_{i,j,1,1}\\
     \vdots\\
  \lambda_{i,w-1}   V_{i,j,w-1,1}
\end{array}\right)  
\end{equation*}
and
\begin{equation*}
\left(\begin{array}{c}
     V_{i,j,0,0}  \\
     V_{i,j,1,0}\\
     \vdots\\
     V_{i,j,w-1,0}
\end{array}\right)A_jA_i=
\lambda_{j}\left(\begin{array}{c}
    V_{i,j,0,1}  \\
     V_{i,j,1,1}\\
     \vdots\\
     V_{i,j,w-1,1}
\end{array}\right)  A_i=
\lambda_{j} \left(\begin{array}{c}
   \lambda_{i,0} V_{i,j,0,1}  \\
    \lambda_{i,1} V_{i,j,1,1}\\
     \vdots\\
    \lambda_{i,w-1} V_{i,j,w-1,1}
\end{array}\right)  
\end{equation*}

{Let $\oplus$ denote the addition modulo $w$,} for $t\in [1:w)$, we have
\begin{equation*}
\left(\begin{array}{c}
     V_{i,j,0,t}  \\
     V_{i,j,1,t}\\
     \vdots\\
     V_{i,j,w-1,t}
\end{array}\right)A_iA_j=
 \left(\begin{array}{c}
    \lambda_{i,0}V_{i,j,0,t}  \\
  \lambda_{i,1}   V_{i,j,1,t}\\
     \vdots\\
  \lambda_{i,w-1}   V_{i,j,w-1,t}
\end{array}\right)A_j=
 \left(\begin{array}{c}
   \lambda_{i,0}V_{i,j,0,t\oplus1}  \\
  \lambda_{i,1}   V_{i,j,1,t\oplus1}\\
     \vdots\\
  \lambda_{i,w-1}   V_{i,j,w-1,t\oplus1}
\end{array}\right)  
\end{equation*}
and
\begin{equation*}
\left(\begin{array}{c}
     V_{i,j,0,t}  \\
     V_{i,j,1,t}\\
     \vdots\\
     V_{i,j,w-1,t}
\end{array}\right)A_jA_i=
\left(\begin{array}{c}
    V_{i,j,0,t\oplus1}  \\
     V_{i,j,1,t\oplus1}\\
     \vdots\\
     V_{i,j,w-1,t\oplus1}
\end{array}\right)  A_i=
 \left(\begin{array}{c}
   \lambda_{i,0} V_{i,j,0,t\oplus1}  \\
    \lambda_{i,1} V_{i,j,1,t\oplus1}\\
     \vdots\\
    \lambda_{i,w-1} V_{i,j,w-1,t\oplus1}
\end{array}\right),  
\end{equation*}
}

The above equations imply that 
\begin{equation*}
A_iA_j=A_jA_i    
\end{equation*}
for $i\in [0:m)$ and $j\in [m:n)$ such that $j\ne i+m$.
\end{proof}

{

\begin{Remark}
For completeness, we show here that $A_i$ does not commute with $A_{i+m}$ for $i\in [0: m)$.  By \eqref{Eqn_matrix_A}, we have
\begin{equation*}
V_{i,0}A_iA_{i+m}=\lambda_{i,0}V_{i,0}A_{i+m}=\lambda_{i,0}\lambda_{i+m} V_{i,1},   
\end{equation*}
and
\begin{equation*}
V_{i,0}A_{i+m}A_i=\lambda_{i+m}V_{i,1}A_i=\lambda_{i,1}\lambda_{i+m} V_{i,1},   
\end{equation*}
which shows that $A_iA_{i+m}\ne A_{i+m}A_i$ whenever $\lambda_{i,0}\ne \lambda_{i,1}$. As will be shown later that optimal repair bandwidth requires that $\lambda_{i,0} \lambda_{i,1},\dots, \lambda_{i,r-1}$ be pairwise distinct.
\end{Remark}
}

\begin{Theorem}\label{ThmMDS_condi}
The $(n,k)$ array code $\mathcal{C}$ defined in Construction \ref{Con1} is MDS over $\mathbf{F}_q$ if $A_i-A_j$ is nonsingular for any distinct $i,j\in [0:n)$.
\end{Theorem}

\begin{proof}
The code is MDS if
\begin{equation*}
\begin{pmatrix}
I&I&\cdots&I\\
A_{i_0} &A_{i_1}&\cdots&A_{i_{r-1}}\\ 
\vdots &\vdots&\ddots&\vdots\\    
A_{i_0}^{r-1} &A_{i_1}^{r-1}&\cdots&A_{i_{r-1}}^{r-1}\\    
\end{pmatrix}    
\end{equation*} is nonsingular for any $r$ integers $0\le i_0<i_1<\cdots<i_{r-1}<n$.   
By Lemma \ref{lem-commute-matrix}, each $A_{i_j}$ commutes with at least $r-2$ matrices in the set $\{A_{i_t}\mid t\in [0: r)\backslash\{j\}\}$ for any $j\in [0: r)$. Applying Lemma \ref{Lemma block matrix} yields the desired result. 
\end{proof}

\begin{Theorem}\label{Thm_diff_A}
For the code $\mathcal{C}$ in Construction \ref{Con1}, $A_i-A_j$ is nonsingular for any distinct $i,j\in [0:n)$ if
\begin{itemize}
\item [i)] $\lambda_{i,u}\ne \lambda_{j,v}$ for $i,j\in [0:m)$ with $i\ne j$, and $u,v\in [0: w)$;
\item  [ii)] $\lambda_{i,t}^w\ne \lambda_j$ for $i\in [0:m)$, $t\in [0: w)$, and $j\in [m: n)\backslash \{i+m\}$;
\item [iii)] $\prod_{t=0}^{w-1}\lambda_{i,t}\ne \lambda_{i+m}$ for $i\in [0: m)$;
\item [iv)] $\lambda_{i}\ne \lambda_{j}$ for $i, j\in [m: n)$ with $i\ne j$.
\end{itemize}
\end{Theorem}
\begin{proof}
\begin{itemize}
\item [i)] If $i,j \in [0: m)$ with $i\ne j$, for {$t\in [0: w)$}, we have that
\begin{align*}
\left(\begin{array}{c}
     V_{i,j,t,0}  \\
     V_{i,j,t,1}\\
     \vdots
\\
     V_{i,j,t,w-1}
\end{array}\right)(A_i-A_j)&= \lambda_{i,t}\left(\begin{array}{c}
     V_{i,j,t,0}  \\
     V_{i,j,t,1}\\
     \vdots
\\
     V_{i,j,t,w-1}
\end{array}\right)- \left(\begin{array}{c}
    \lambda_{j,0} V_{i,j,t,0}  \\
    \lambda_{j,1} V_{i,j,t,1}\\
     \vdots
\\
     \lambda_{j,w-1}V_{i,j,t,w-1}
\end{array}\right).
\end{align*}
By the definition of $V_{i,j, t,t'}$ where $t,t'\in [0: w)$, we clearly have that $A_i-A_j$ is nonsingular if and only if $\lambda_{i,u}\ne \lambda_{j,v}$ for $u,v\in [0: w)$.

    \item [ii)]
If $i\in [0: m)$ and $j\in [m: n)\backslash\{i+m\}$, for {$t\in [0: w)$}, we have
\begin{align*}
\left(\begin{array}{c}
     V_{i,j,t,0}  \\
     V_{i,j,t,1}\\
     \vdots
\\
     V_{i,j,t,w-1}
\end{array}\right)(A_i-A_j)&= \lambda_{i,t}\left(\begin{array}{c}
     V_{i,j,t,0}  \\
     V_{i,j,t,1}\\
     \vdots
\\
     V_{i,j,t,w-1}
\end{array}\right)- \left(\begin{array}{c}
    \lambda_j V_{i,j,t,1}  \\
     V_{i,j,t,2}\\
     \vdots\\
     V_{i,j,t,0}
\end{array}\right)\\&=\begin{pmatrix}
\lambda_{i,t}I&-\lambda_jI&&&\\
&\lambda_{i,t}I&-I&\\
&&\ddots&\ddots&\\
&&&\lambda_{i,t}I&-I\\
-I&&&&\lambda_{i,t}I
\end{pmatrix}  \left(\begin{array}{c}
     V_{i,j,t,0}  \\
     V_{i,j,t,1}\\
     \vdots\\
     V_{i,j,t,w-2}\\
V_{i,j,t,w-1}
\end{array}\right).
\end{align*}
By the definition of $V_{i,j, t,t'}$, we have that $A_i-A_j$ is nonsingular if and only if $\lambda_{i,t}^w\ne \lambda_j$.

  \item [iii)]
If $i\in [0: m)$ and $j=i+m$, we have
\begin{equation*}
\begin{pmatrix}
V_{i,0}\\
V_{i,1}\\
\vdots\\
V_{i,w-1}
\end{pmatrix}(A_{i}-A_{j})=\begin{pmatrix}
\lambda_{i,0}V_{i,0}\\
\lambda_{i,1}V_{i,1}\\
\vdots\\
\lambda_{i,w-1}V_{i,w-1}
\end{pmatrix}-\begin{pmatrix}
\lambda_{i+m}V_{i,1}\\
V_{i,2}\\
\vdots\\
V_{i,0}
\end{pmatrix}=\begin{pmatrix}
\lambda_{i,0}I&-\lambda_jI&&&\\
&\lambda_{i,1}I&-I&\\
&&\ddots&\ddots&\\
&&&\lambda_{i,w-2}I&-I\\
-I&&&&\lambda_{i,w-1}I
\end{pmatrix}\begin{pmatrix}
V_{i,0}\\
V_{i,1}\\
\vdots\\
V_{i,w-1}
\end{pmatrix},
\end{equation*}
which indicates that $A_i-A_j$ is nonsingular if
$\prod_{t=0}^{w-1}\lambda_{i,t}\ne \lambda_j$.

\item [iv)] If $i,j \in [m: n)$ with $i\ne j$, we have that $A_i-A_j$ is nonsingular if and only if $\lambda_{i}\ne \lambda_{j}$ according to the proof in \cite{Barg1}.
\end{itemize}    
\end{proof}

Combining Theorems \ref{ThmMDS_condi} and \ref{Thm_diff_A} yields Theorem \ref{Thm_MDS}.

\section{Repair Property of the Array Code in Construction \ref{Con1}: Proof of Theorem \ref{Thm_opt_RB}}\label{sec:repair}
\subsection{Proof of Theorem \ref{Thm_opt_RB} when $d=n-1$}
When $d=n-1$, i.e., $w=r$, it suffices to show that conditions R1 and R2 outlined in Section \ref{sec:repair_mech} hold for the code in Construction \ref{Con1}.
Define the selection matrices as
\begin{equation}\label{Eqn_S}
S_{i}\hspace{-1mm}=\hspace{-1mm}\left\{\hspace{-2mm}\begin{array}{ll} \xi_{i,0}V_{i,0}+\xi_{i,1}V_{i,1}+\cdots+\xi_{i,r-1}V_{i,r-1}, &\hspace{-2mm}i\in [0: m),\\[2pt] 
V_{i,0}, &\hspace{-2mm}i\in [m:2m), \end{array}\right.   
\end{equation} 
for some $\xi_{i,u}\in\mathbf{F}_q\backslash\{0\}$.

\begin{Proposition}\label{prop1}
When $w=r$, condition R1 holds for the code in Construction \ref{Con1} if $\lambda_{i,0}, \lambda_{i,1},\dots, \lambda_{i,r-1}$ are pairwise distinct.    
\end{Proposition}
\begin{proof}
{For any $i\in [0:n)$} and $t\in [0:r)$, we have

\begin{align*} 
S_{i}A_{t,i} =&\left\{\begin{array}{ll} \xi_{i,0}\lambda_{i,0}^t V_{i,0}+\xi_{i,1}\lambda_{i,1}^t V_{i,1}+\cdots+\xi_{i,r-1}\lambda_{i,r-1}^t V_{i,r-1}, &i\in [0: m),\\[4pt] 
V_{i,0}, &i\in [m:2m) \text{ and }t=0,\\[4pt] 
\lambda_iV_{i,t}, &i\in [m:2m) \text{ and }t\in [1: r). \end{array}\right. \end{align*} 

Thus, for $i\in [0:m)$, we have
\begin{align*}
\operatorname{rank}\begin{pmatrix}
S_i\\
S_iA_{1,i}\\
\vdots\\
S_iA_{r-1,i}
\end{pmatrix}&=\operatorname{rank}\begin{pmatrix}
\xi_{i,0}V_{i,0}+\xi_{i,1}V_{i,1}+\cdots+\xi_{i,r-1} V_{i,r-1}\\
\xi_{i,0}\lambda_{i,0} V_{i,0}+\xi_{i,1}\lambda_{i,1} V_{i,1}+\cdots+\xi_{i,r-1}\lambda_{i,r-1} V_{i,r-1}\\
\vdots\\
\xi_{i,0}\lambda_{i,0}^{r-1} V_{i,0}+\xi_{i,1}\lambda_{i,1}^{r-1} V_{i,1}+\cdots+\xi_{i,r-1}\lambda_{i,r-1}^{r-1} V_{i,r-1}\\
\end{pmatrix}\\
&=\operatorname{rank}\begin{pmatrix}
I_{\frac{N}{r}}&I_{\frac{N}{r}}& \cdots &I_{\frac{N}{r}}\\
\lambda_{i,0}I_{\frac{N}{r}} &\lambda_{i,1}I_{\frac{N}{r}} &\cdots&\lambda_{i,r-1}I_{\frac{N}{r}}\\
\vdots&\vdots&\ddots&\vdots\\
\lambda_{i,0}^{r-1}I_{\frac{N}{r}}&\lambda_{i,1}^{r-1} I_{\frac{N}{r}}&\cdots&\lambda_{i,r-1}^{r-1}I_{\frac{N}{r}}\\
\end{pmatrix}\begin{pmatrix}
\xi_{i,0}V_{i,0}\\
\xi_{i,1}V_{i,1}\\
\vdots\\
\xi_{i,r-1}V_{i,r-1}
\end{pmatrix}\\
&=N
\end{align*}
since $\lambda_{i,0},\lambda_{i,1},\ldots,\lambda_{i,r-1}$ are pairwise distinct.

For $i\in [m:n)$, we have 
\begin{equation*}
\operatorname{rank}\begin{pmatrix}
S_i\\
S_iA_{1,i}\\
\vdots\\
S_iA_{r-1,i}
\end{pmatrix}=\operatorname{rank}\begin{pmatrix}
V_{i,0}\\
\lambda_{i}V_{i,1} \\
\vdots\\
\lambda_{i} V_{i,r-1}\\
\end{pmatrix}=N.
\end{equation*}    
\end{proof}

\begin{Proposition}\label{prop2}
When $d=n-1$, condition R2 holds for the code in Construction \ref{Con1} if { there exists $\mu_i \in \mathbf{F}_q$ such that $\lambda_{i+m}=\mu_i^w$ for $i\in [0: m)$.}
\end{Proposition}
\begin{proof}
{Let $\mu_i$ be some element in $\mathbf{F}_q$ and $\lambda_{i+m}=\mu_i^w$, define \begin{equation}\label{Eqn_xi}
\xi_{i,u}=\left\{\begin{array}{ll} 1, &u=0,\\
\mu_i^{w-u}, &u\in [1: w).\end{array}\right.    
\end{equation}}
Note that $w=r$ according to $d=n-1$. We prove the statement according to the following four cases.
\begin{itemize}
 \item For $i,j\in [0: m)$ with $i\ne j$,
\begin{equation*}
\operatorname{rank}\begin{pmatrix}
S_i\\
S_iA_j^t
\end{pmatrix}
=\operatorname{rank}\begin{pmatrix}
\sum\limits_{u=0}^{r-1}\xi_{i,u}V_{i,u}\\
\sum\limits_{u=0}^{r-1}\xi_{i,u}V_{i,j,u,0}A_j^t\\
\sum\limits_{u=0}^{r-1}\xi_{i,u}V_{i,j,u,1}A_j^t\\
\vdots\\
\sum\limits_{u=0}^{r-1}\xi_{i,u}V_{i,j,u,r-1}A_j^t
\end{pmatrix}=\operatorname{rank}\begin{pmatrix}
\sum\limits_{u=0}^{r-1}\xi_{i,u}V_{i,u}\\
\lambda_{j,0}^t\sum\limits_{u=0}^{r-1}\xi_{i,u}V_{i,j,u,0}\\
\lambda_{j,1}^t\sum\limits_{u=0}^{r-1}\xi_{i,u}V_{i,j,u,1}\\
\vdots\\
\lambda_{j,r-1}^t\sum\limits_{u=0}^{r-1}\xi_{i,u}V_{i,j,u,r-1}
\end{pmatrix}=\operatorname{rank}(S_i).  
\end{equation*}

\item For $i\in [0: m)$ and $j\in [m: n)$, define
\begin{equation*}
\delta_u=\left\{\begin{array}{ll} 1, &u=1,\\
0, &\text{otherwise}.\end{array}\right. 
\end{equation*}
If $j\ne i+m$, then
\begin{align*}
\operatorname{rank}\begin{pmatrix}
S_i\\
S_iA_j^t
\end{pmatrix}
=&\operatorname{rank}\begin{pmatrix}
\sum\limits_{u=0}^{r-1}\xi_{i,u} V_{i,u}\\
\sum\limits_{u=0}^{r-1}\xi_{i,u} V_{i,u}A_j^t
\end{pmatrix}\\
=&\operatorname{rank}\begin{pmatrix}
\sum\limits_{u=0}^{r-1}\xi_{i,u} V_{i,u}\\
\sum\limits_{u=0}^{r-1}\xi_{i,u}V_{i,j,u,0}A_j^t\\
\sum\limits_{u=0}^{r-1}\xi_{i,u}V_{i,j,u,1}A_j^t\\
\vdots\\
\sum\limits_{u=0}^{r-1}\xi_{i,u}V_{i,j,u,r-1}A_j^t
\end{pmatrix} \\
=&\operatorname{rank}\begin{pmatrix}
\sum\limits_{u=0}^{r-1}\xi_{i,u} V_{i,u}\\
\lambda_j^{\left(\sum\limits_{v=1}^t\delta_{v}\right)}\sum\limits_{u=0}^{r-1}\xi_{i,u}V_{i,j,u,t}\\
\lambda_j^{\left(\sum\limits_{v=1}^t\delta_{1\oplus v}\right)}\sum\limits_{u=0}^{r-1}\xi_{i,u}V_{i,j,u,t\oplus1}\\
\vdots\\
\lambda_j^{\left(\sum\limits_{v=1}^t\delta_{(r-1)\oplus v}\right)}\sum\limits_{u=0}^{r-1}\xi_{i,u}V_{i,j,u,t\oplus (r-1)}
\end{pmatrix}\\
=&\operatorname{rank}(S_i)
\end{align*}
by \eqref{Eqn_matrix_A};
Otherwise,
we have
\begin{align*}
&\operatorname{rank}\begin{pmatrix}
S_i\\
S_iA_j^t
\end{pmatrix}\\
&=\operatorname{rank}\begin{pmatrix}
\sum\limits_{u=0}^{r-1}\xi_{i,u} V_{i,u}\\
\sum\limits_{u=0}^{r-1}\xi_{i,u} V_{j,u}A_j^t
\end{pmatrix}\\
&=\operatorname{rank}\begin{pmatrix}
\sum\limits_{u=0}^{r-1}\xi_{i,u}V_{i,u}\\
\sum\limits_{u=0}^{r-1}\xi_{i,u}\lambda_j^{\sum\limits_{v=1}^t\delta_{u\oplus v}}V_{j,u\oplus t}
\end{pmatrix}\\
&=\operatorname{rank}\begin{pmatrix}
\sum\limits_{u=0}^{r-1}\xi_{i,u}V_{i,u}\\
\xi_{i,0}\lambda_j^{\sum\limits_{v=1}^t\delta_{v}}V_{j,t}+\xi_{i,1}\lambda_j^{\sum\limits_{v=1}^t\delta_{1\oplus v}}V_{j,t\oplus 1}+\cdots+\xi_{i,r-1}\lambda_j^{\sum\limits_{v=1}^t\delta_{(r-1)\oplus v}}V_{j,t\oplus (r-1)}
\end{pmatrix}\\
&=\operatorname{rank}\begin{pmatrix}
\xi_{i,0}V_{i,0}+\cdots+\xi_{i,t-1}V_{i,t-1} +\xi_{i,t}V_{i,t} +\cdots +{\xi_{i,r-1}V_{i,r-1}}\\
\xi_{i,r-t}\lambda_j^{\sum\limits_{v=1}^t\delta_{(r-t)\oplus v}}V_{i,0} +\cdots+ \xi_{i,r-1}\lambda_j^{\sum\limits_{v=1}^t\delta_{(r-1)\oplus v}}V_{i,t-1}  + \xi_{i,0}\lambda_j^{\sum\limits_{v=1}^t\delta_{v}}V_{i,t}+\cdots +\xi_{i,r-t-1}\lambda_j^{\sum\limits_{v=1}^t\delta_{(r-t-1)\oplus v}}V_{i,r-1} \end{pmatrix}\\
&=\operatorname{rank}\begin{pmatrix}
\xi_{i,0}V_{i,0}+\xi_{i,1}V_{i,1}+\cdots+\xi_{i,t-1}V_{i,t-1} +\xi_{i,t}V_{i,t} +\xi_{i,t+1}V_{i,t+1} +\cdots +\xi_{i,r-1}V_{i,r-1}\\
\xi_{i,r-t}V_{i,0}+ \xi_{i,r-t+1}\lambda_jV_{i,1}+ \cdots +\xi_{i,r-1}\lambda_j V_{i,t-1} + \xi_{i,0}\lambda_jV_{i,t} +\xi_{i,1}V_{i,t+1}+\cdots + \xi_{i,r-t-1}V_{i,r-1}\end{pmatrix}\\
&=\operatorname{rank}(S_i)
\end{align*}
since
{
\begin{equation*}
\frac{\xi_{i,r-t}}{\xi_{i,0}}=\mu_{i}^{t}
,~
\frac{\xi_{i,r-t+u}\lambda_{j}}{\xi_{i,u}}=\frac{\mu_{i}^{t-u+r}}{\mu_{i}^{r-u}}=\mu_{i}^{t}\text{ for }u\in [1: t),
\end{equation*}
and
\begin{equation*}
\frac{\xi_{i,0}\lambda_j}{\xi_{i,t}}=\frac{\mu_{i}^r}{\mu_{i}^{r-t}}=\mu_{i}^{t}
,~
\frac{\xi_{i,u}}{\xi_{i,t+u}}=\frac{\mu_{i}^{r-u}}{\mu_{i}^{r-t-u}}=\mu_{i}^{t}\text{ for }u\in [1: r-t)
\end{equation*}}
according to \eqref{Eqn_xi}.

\item For $i,j\in [m: n)$ with $i\ne j$,
\begin{equation*}
\operatorname{rank}\begin{pmatrix}
S_i\\
S_iA_j^t
\end{pmatrix}
=\operatorname{rank}\begin{pmatrix}
V_{i,0}\\
V_{i,j,0,0}A_j^t\\
V_{i,j,0,1}A_j^t\\
\vdots\\
V_{i,j,0,r-1}A_j^t
\end{pmatrix}=\operatorname{rank}\begin{pmatrix}
V_{i,0}\\
\lambda_j^{\left(\sum\limits_{v=1}^t\delta_{v}\right)}V_{i,j,0,t}\\
\lambda_j^{\left(\sum\limits_{v=1}^t\delta_{1\oplus v}\right)}V_{i,j,0,t\oplus1}\\
\vdots\\
\lambda_j^{\left(\sum\limits_{v=1}^t\delta_{(r-1)\oplus v}\right)}V_{i,j,0,t\oplus (r-1)}
\end{pmatrix}=\operatorname{rank}(S_i).  
\end{equation*}

\item For $i\in [m: n)$ and $j\in [0: m)$, if $j=i-m$,
we have
\begin{equation*}
\operatorname{rank}\begin{pmatrix}
S_i\\
S_iA_j^t
\end{pmatrix}
=\operatorname{rank}\begin{pmatrix}
V_{i,0}\\
V_{j,0}A_j^t
\end{pmatrix}=\operatorname{rank}\begin{pmatrix}
V_{i,0}\\
\lambda_{j,0}^tV_{i,0}
\end{pmatrix}=\operatorname{rank}(S_i);  
\end{equation*}
else
\begin{equation*}
\operatorname{rank}\begin{pmatrix}
S_i\\
S_iA_j^t
\end{pmatrix}
=\operatorname{rank}\begin{pmatrix}
V_{i,0}\\
V_{i,j,0,0}A_j^t\\
V_{i,j,0,1}A_j^t\\
\vdots\\
V_{i,j,0,r-1}A_j^t
\end{pmatrix}=\operatorname{rank}\begin{pmatrix}
V_{i,0}\\
\lambda_{j,0}^tV_{i,j,0,0}\\
\lambda_{j,1}^tV_{i,j,0,1}\\
\vdots\\
\lambda_{j,r-1}^tV_{i,j,0,r-1}\\
\end{pmatrix}=\operatorname{rank}(S_i).  
\end{equation*}
\end{itemize}
\end{proof}

By Propositions \ref{prop1} and \ref{prop2}, Theorem \ref{Thm_opt_RB} follows for $d=n-1$.

\subsection{Proof of Theorem \ref{Thm_opt_RB} when $d<n-1$}
In this case, there are $r-w$ nodes that are not contacted during node repair given the repair degree $d=k+w-1<n-1$. We prove only the case $r-w=1$, i.e., $d=n-2$; the case $r-w>1$ follows similarly.

Without loss of generality, suppose node $0$ fails and node $n-1$ is not contacted. From \eqref{Eqn parity check eq} and \eqref{Eqn_A_ti}, we have
\begin{equation*}
\left\{\begin{array}{l}
\mathbf{f}_0+\mathbf{f}_1+\cdots+\mathbf{f}_{n-1}=\mathbf{0},\\
A_0\mathbf{f}_0+A_1\mathbf{f}_1+\cdots+A_{n-1}\mathbf{f}_{n-1}=\mathbf{0},\\
\hspace{2.5cm}\vdots\\
A_0^{r-1}\mathbf{f}_0+A_1^{r-1}\mathbf{f}_1+\cdots+A_{n-1}^{r-1}\mathbf{f}_{n-1}=\mathbf{0}.
\end{array}\right. 
\end{equation*}
{For $t=r-1, r-2, \ldots, 1$, multiply the $(t-1)$-th parity-check group by $A_{n-1}$ and subtract it from the $t$-th parity-check group to obtain}
\begin{equation*}
\left\{\begin{array}{l}
(A_0-A_{n-1})\mathbf{f}_0+(A_1-A_{n-1})\mathbf{f}_1+\cdots+(A_{n-2}-A_{n-1})\mathbf{f}_{n-2}=\mathbf{0},\\
(A_0-A_{n-1})A_0\mathbf{f}_0+(A_1-A_{n-1})A_1\mathbf{f}_1+\cdots+(A_{n-2}-A_{n-1})A_{n-2}\mathbf{f}_{n-2}=\mathbf{0},\\
\hspace{5cm}\vdots\\
(A_0-A_{n-1})A_0^{r-2}\mathbf{f}_0+(A_1-A_{n-1})A_1^{r-2}\mathbf{f}_1+\cdots+(A_{n-2}-A_{n-1})A_{n-2}^{r-2}\mathbf{f}_{n-2}=\mathbf{0}.
\end{array}\right. 
\end{equation*}
Since $A_{n-1}$ commutes with $A_j$ for $j\in [0: n-1)\backslash\{m-1\}$ by Lemma \ref{lem-commute-matrix}, we have
\begin{equation*}
\left\{\begin{array}{l}
\sum_{j=0,j\ne m-1}^{n-2}(A_j-A_{n-1})\mathbf{f}_j+(A_{m-1}-A_{n-1})\mathbf{f}_{m-1}=\mathbf{0},\\
\sum_{j=0,j\ne m-1}^{n-2}A_j(A_j-A_{n-1})\mathbf{f}_j+(A_{m-1}-A_{n-1})A_{m-1}\mathbf{f}_{m-1}=\mathbf{0},\\
\hspace{5cm}\vdots\\
\sum_{j=0,j\ne m-1}^{n-2}A_j^{r-2}(A_j-A_{n-1})\mathbf{f}_j+(A_{m-1}-A_{n-1})A_{m-1}^{r-2}\mathbf{f}_{m-1}=\mathbf{0}.
\end{array}\right. 
\end{equation*}
Multiplying each equation by $(A_{m-1}-A_{n-1})^{-1}$ from the left yields
\begin{equation*}
\left\{\begin{array}{l}
\sum_{j=0,j\ne m-1}^{n-2}(A_{m-1}-A_{n-1})^{-1}(A_j-A_{n-1})\mathbf{f}_j+\mathbf{f}_{m-1}=\mathbf{0},\\
\sum_{j=0,j\ne m-1}^{n-2}A_j(A_{m-1}-A_{n-1})^{-1}(A_j-A_{n-1})\mathbf{f}_j+A_{m-1}\mathbf{f}_{m-1}=\mathbf{0},\\
\hspace{5cm}\vdots\\
\sum_{j=0,j\ne m-1}^{n-2}A_j^{r-2}(A_{m-1}-A_{n-1})^{-1}(A_j-A_{n-1})\mathbf{f}_j+A_{m-1}^{r-2}\mathbf{f}_{m-1}=\mathbf{0},
\end{array}\right. 
\end{equation*}
where we used the commutativity of $A_{m-1}$ and $A_{n-1}$ with $A_j$ for $j\ne m-1$. 
That is,
\begin{align*}
\begin{pmatrix}
I&I&\cdots &I\\
A_0&A_1&\cdots &A_{n-2}\\
\vdots & \vdots & \ddots & \vdots \\
A_0^{r-2}&A_1^{r-2}&\cdots &A_{n-2}^{r-2}
\end{pmatrix}
\begin{pmatrix}
(A_{m-1}-A_{n-1})^{-1}(A_0-A_{n-1})\mathbf{f}_0\\
\vdots\\
(A_{m-1}-A_{n-1})^{-1}(A_{m-2}-A_{n-1})\mathbf{f}_{m-2}\\
\mathbf{f}_{m-1}\\
(A_{m-1}-A_{n-1})^{-1}(A_{m+1}-A_{n-1})\mathbf{f}_{m+1}\\
\vdots\\
(A_{m-1}-A_{n-1})^{-1}(A_{n-2}-A_{n-1})\mathbf{f}_{n-2}
\end{pmatrix}=\bf{0},
\end{align*}which defines an $(n'=n-1, k'=k)$ array code $\mathcal{C}'$ with parity-check matrix $(A_{t,i})_{t\in [0: r-1),i\in[0: n-1)}$ and codeword $(\mathbf{f}_0',\dots,\mathbf{f}_{n-2}')$ where
\begin{equation*}
\mathbf{f}_j'=\left\{\begin{array}{ll}
  (A_{m-1}-A_{n-1})^{-1}(A_j-A_{n-1})\mathbf{f}_j, & j\in [0: n-1)\backslash\{m-1\}, \\
   \mathbf{f}_j, & j=m-1.
\end{array}\right.
\end{equation*}
The code $\mathcal{C}'$ has repair degree $d'=d=n'-1$ and corresponds to a shortened version of an $(n, k+1)$ code from Construction \ref{Con1}.

Thus, repairing $(A_{m-1}-A_{n-1})^{-1}(A_0-A_{n-1})\mathbf{f}_0$ requires downloading $N/r'$ symbols from $\mathbf{f}_j'$ for $j\in [1: n-1)$, which is equivalent to repairing $\mathbf{f}_0$ by downloading $N/r'=N/w$ symbols from $\mathbf{f}_j$ for $j\in [1: n-1)$. The repair bandwidth is therefore
\begin{equation*}
 (n-2)\frac{N}{w}=\frac{dN}{d-k+1},
\end{equation*}
achieving the lower bound in \eqref{Eqn_gamma_opt}. This completes the proof of Theorem \ref{Thm_opt_RB} for $d<n-1$.

\section{Comparisons}\label{sec:comp}
In this section, we compare key parameters between the proposed MSR code and existing MSR codes that support $d<n-1$, as summarized in Table \ref{Table comp MSR}.

\begin{table*}[htbp]
\begin{center}
\caption{Comparison of key parameters of $(n,k)$ MSR codes supporting $d<n-1$ with sub-packetization level $N$; where existing codes are named using the initials of the authors' surnames.}
\label{Table comp MSR}
\setlength{\tabcolsep}{2.3pt}
\begin{tabular}{|c|c|c|c|c|c|}
\hline  
& $N$& Field size  & Repair degree  & Density &  References \\
\hline
 YB code 1    &$w^{n}$   &  $q\ge wn$  &  $d=k+w-1\in [k+1:n)$  & $\frac{1}{N}$ &  \cite{Barg1} \\ 
\hline
YB code 2      &$\left\{\hspace{-2mm}\begin{array}{ll}
w^{n-1}, &\hspace{-2mm}w=r\\
w^n, &\hspace{-2mm}w<r
\end{array} \right.$   &  $q>n$  &    $d=k+w-1\in [k+1:n)$ & $\frac{1}{N}$ & \cite{Barg1}\\
\hline
\multirow{2}{*}{VBK code}    &$w^{\lceil \frac{n}{w} \rceil}$   &\multirow{2}{*}{  $q\ge \left\{\begin{array}{ll}
6\lceil \frac{n}{2} \rceil+2, &{\rm if}~w=2\\
18\lceil \frac{n}{w} \rceil+2, &{\rm if}~w=3,4
\end{array} \right.$}  &   \multirow{2}{*}{$d=k+1, k+2, k+3$} & \multirow{2}{*}{$\frac{2N-1}{N^2}$}& \multirow{2}{*}{\cite{vajha2023small}}\\ 
&($w=2,3,4$)&&&&\\[4pt]
\hline
CB code     &$w^{n}$   &  $q>n+w$ & $d=k+w-1\in [k+1:n)$  &   $\frac{2N-1}{N^2}$&  \cite{chen2020explicit}\\
\hline
LLT code    &$w^{\lceil \frac{n}{2} \rceil}$  &$q>n+\lceil\frac{n}{2} \rceil w$  &  $d=k+w-1\in [k+1:n)$  & $\frac{2N-1}{N^2}$&   \cite{liu2022generic}\\
\hline
ZZ code  &$w^{\lceil \frac{n}{2} \rceil}$ &    $q\ge wn$   &  $d=k+w-1\in [k+1:n)$  & $\frac{3N-1}{2N^2}$&\cite{zhang2023vertical}  \\
\hline
LLTHBZ code   &$w^{\lceil \frac{n}{2} \rceil}$    &  $q>\left\{\hspace{-2mm}\begin{array}{ll}
4\lceil \frac{n}{2}\rceil, &{\rm if}~w=2\\
\lceil \frac{n}{2}\rceil (w+1), &{\rm if}~w\in [3: r)\\
\lceil \frac{n}{2}\rceil w, &{\rm if}~w=r
\end{array} \right.$  &  $d=k+w-1\in [k+1:n)$  &$\frac{3N-1}{2N^2}$& \cite{li2023MDS}\\
\hline
WLHY code  &$w^{\lceil \frac{n}{3} \rceil}$    &  $q>2n$  &  $d=k+1$  & $\frac{4}{3N}$& \cite{wang2023constructing} \\
\hline
LWHY code 1     &$w^{\lceil \frac{n}{w} \rceil}$    &  $q>\left\{\hspace{-2mm}\begin{array}{ll}
nw+(w-1)2^{w-2}, &{\rm implicit~construction}\\
\max\{nw,p^{w^{5}}\}, &{\rm explicit~construction}
\end{array} \right.$  &  $d=k+w-1\in [k+1:n)$ & $\frac{2N-1}{N^2}$ & \cite{li2024msr}\\\hline
LWHY code 2    &$w^{\lceil \frac{n}{w+1} \rceil}$    &  $q>\left\{\hspace{-2mm}\begin{array}{ll}
nw+w2^{w-1}, &{\rm implicit~construction}\\
\max\{nw,p^{w^{5}}\}, &{\rm explicit~construction}
\end{array} \right.$  &  $d=k+w-1\in [k+1:n)$ &$\frac{(2w-1)N-w+1}{wN^2}$ & \cite{li2024msr}\\\hline
New MSR code    &$w^{\lceil \frac{n}{2} \rceil}$    &  $q>\lceil \frac{n}{2}\rceil w$ {$ {\rm with} \gcd(w,q-1) = 1$ } &  $d=k+w-1\in [k+1:n)$ &$\frac{1}{N}$  & Constr. \ref{Con1}\\
\hline
\end{tabular}
\end{center}
\end{table*}

From Table \ref{Table comp MSR}, {we observe that the new MSR code requires the field size $q$ to satisfy $\gcd(w,q-1) = 1$, whereas none of the other MSR codes listed in Table II imposes this condition. This restriction narrows the range of parameters for which the proposed code has a field-size advantage over the YB code 1 \cite{Barg1}, VBK code \cite{vajha2023small}, LLT code \cite{liu2022generic}, and the ZZ code \cite{zhang2023vertical}.} The detailed  advantages of the new MSR code are as follows:
\begin{itemize}
\item [i)] It achieves a significantly smaller sub-packetization level than the YB codes 1, 2 in \cite{Barg1} and the CB code in \cite{chen2020explicit}, and requires a smaller finite field than YB code 1 {for some code parameters $n$ and $w$, namely those for which there exists a prime power $q$ such that $\lceil \frac{n}{2}\rceil w<q<wn$ and $\gcd(w,q-1) = 1$}.
\item [ii)] It supports any repair degree $d\in [k+1 : n)$, offering greater flexibility than the VBK code in \cite{vajha2023small} (restricted to $d=k+1, k+2, k+3$) and the WLHY code in \cite{wang2023constructing} (restricted to $d=k+1$). Moreover, it requires a much smaller finite field than the VBK code when $w=2$ (i.e., under the same sub-packetization level) { for some code length $n$, namely that for which there exists a prime power $q$ such that $\lceil \frac{n}{2}\rceil w<q<6\lceil \frac{n}{2} \rceil+2$ and $\gcd(w,q-1) = 1$}.
\item [iii)] While it shares the same sub-packetization level as the LLT code in \cite{liu2022generic} and the ZZ code in \cite{zhang2023vertical}, it demands a smaller finite field { for some code parameters $n$ and $w$, namely those for which there exists a prime power $q$ such that $\lceil \frac{n}{2}\rceil w<q<wn$ and $\gcd(w,q-1) = 1$}.
\item [iv)] It subsumes the YB codes in \cite{Barg1}, which can be obtained by shortening the new code.
\item [v)] Although the recent LWHY codes 1 and 2 in \cite{li2024msr} have smaller sub-packetization levels, their explicit versions generally require much larger finite fields (e.g., $p^{w^5}$) compared to $nw$ for a wide range of parameters.
\item [vi)] The new MSR code achieves the smallest density among all existing MSR codes listed in Table \ref{Table comp MSR}, matching that of the Ye–Barg codes. This sparsity directly translates to reduced computational overhead in encoding and decoding, enhancing practical feasibility.
\end{itemize}

\section{Conclusion}\label{sec:conclusion}
We proposed an $(n, k)$ MSR code constructed from a block Vandermonde parity-check matrix that employs diagonal matrices together with generalized permutation matrices as building blocks. The proposed MSR code achieves a small sub-packetization level of $(d-k+1)^{\lceil \frac{n}{2}\rceil}$ and supports any repair degree $d\in [k+1: n)$. The resulting parity-check matrix is sparse, which directly translates to reduced computational overhead during encoding and decoding, enhancing practical feasibility. This work not only broadens the family of Ye–Barg codes but also introduces greater flexibility into their design framework. By relaxing structural constraints while preserving the MDS property, it opens new avenues for constructing MSR codes for future large-scale storage systems.


\begin{thebibliography}{99}

\bibitem{dimakis2010network} A. G. Dimakis, P. Godfrey, Y. Wu, M. Wainwright, and K. Ramchandran, ``Network coding for distributed storage systems," \textit{IEEE Trans. Inform. Theory,} vol. 56, no. 9, pp. 4539-4551, Sep. 2010.

\bibitem{Rashmi2011optimal} K. V. Rashmi, N. B. Shah, and P. V. Kumar, ``Optimal exact-regenerating codes for distributed storage at the {MSR} and {MBR} points via a product-matrix construction," \textit{IEEE Trans. Inform. Theory,} vol. 57, no. 8, pp. 5227-5239, Aug. 2011.

\bibitem{wang2012long}  Z. Wang, T. Tamo, and J. Bruck, ``Long MDS Codes for Optimal Repair
Bandwidth," \textit{2012 IEEE International Symposium on Information Theory (ISIT)}, pp. 1182–1186, IEEE, 2012.

\bibitem{Zigzag} T. Tamo, Z. Wang, and J. Bruck, ``Zigzag codes: MDS array codes
with optimal rebuilding," \textit{IEEE Trans. Inform. Theory,} vol. 59, no. 3, pp. 1597-1616, Mar. 2013.

\bibitem{hadamard} D.S. Papailiopoulos, A.G. Dimakis, and V.R. Cadambe, ``Repair optimal erasure codes through hadamard designs," \textit{IEEE Trans. Inform. Theory,} vol. 59, no. 5, pp. 3021-3037, May 2013.

\bibitem{Hadamard strategy} X. Tang, B. Yang, J. Li, and H.D.L. Hollmann, ``A new repair strategy for the Hadamard minimum storage regenerating codes for distributed storage systems,'' \textit{IEEE Trans. Inform. Theory,} vol. 61, no. 10, pp. 5271-5279, Oct. 2015.

\bibitem{invariant subspace} J. Li, X. Tang, and U. Parampalli, ``A framework of constructions of minimal storage regenerating codes
with the optimal access/update property," \textit{IEEE Trans. Inform. Theory,} vol. 61, no. 4, pp. 1920-1932, Apr. 2015.


\bibitem{Long IT} Z. Wang, T. Tamo, and J. Bruck, ``Explicit minimum storage regenerating codes," \textit{IEEE Trans. Inform. Theory,} vol. 62, no. 8, pp. 4466-4480, Aug. 2016.

\bibitem{tian2014characterizing} C. Tian, ``Characterizing the rate region of the (4, 3, 3) exact-repair regenerating codes," \textit{IEEE J. Sel. Areas Commun.,} vol. 32, no. 5, pp. 967-975, Dec. 2014.



\bibitem{han2015update} Y. S. Han, H.-T. Pai, R. Zheng, P. K. Varshney ``Update-efficient error-correcting product-matrix codes," \textit{IEEE Trans. Commun.,} vol. 63, no. 6, pp. 1925-1938, June 2015.




\bibitem{Sasidharan-Kumar2} B. Sasidharan, M. Vajha, and P.V. Kumar, ``An explicit, coupled-layer construction of a high-rate MSR code with low sub-packetization level, small field size and all-node repair," \textit{ arXiv: 1607.07335 [cs.IT]}

\bibitem{Goparaju} S. Goparaju, A. Fazeli, and A. Vardy, ``Minimum storage regenerating codes
for all parameters," \textit{IEEE Trans. Inform. Theory,} vol. 63, no. 10, pp. 6318-6328, Oct. 2017.

\bibitem{Barg1} M. Ye and A. Barg, ``Explicit constructions of high-rate MDS array codes with optimal repair bandwidth," \textit{IEEE Trans. Inform. Theory,} vol. 63, no. 4, pp. 2001-2014, Apr. 2017.

\bibitem{Barg2} M. Ye and A. Barg, ``Explicit constructions of optimal-access MDS codes with nearly optimal sub-packetization," \textit{IEEE Trans. Inform. Theory,} vol. 63, no. 10, pp. 6307-6317, Oct. 2017.

\bibitem{li2017generic} J. Li, X. Tang, and C. Tian, ``A generic transformation for optimal repair bandwidth and rebuilding access in MDS codes," in \textit{Proc. IEEE Int. Symp. Inf. Theory (ISIT)}, June 2017,
pp. 1623-1627.

\bibitem{li2018generic} J. Li, X. Tang, and C. Tian, ``A generic transformation to enable optimal repair in MDS codes for distributed storage systems," \textit{IEEE Trans. Inform. Theory,} vol. 64, no. 9, pp. 6257-6267, Sept. 2018.

\bibitem{hou2020binary} H. Hou and P.P.C. Lee, ``Binary MDS array codes with optimal repair," \textit{IEEE Trans. Inform. Theory,} vol. 66, no. 3, pp. 1405--1422, Mar. 2020.

\bibitem{elyasi2020cascade} M. Elyasi and S. Mohajer, ``Cascade codes for distributed storage systems," \textit{IEEE Trans. Inform. Theory,} vol. 66, no. 12, pp. 7490--7527, Dec. 2020.

\bibitem{chen2020explicit} Z. Chen and A. Barg, ``Explicit constructions of MSR codes for clustered distributed storage: The rack-aware storage model," \textit{IEEE Trans. Inform. Theory,} vol. 66, no. 2, pp. 886--899, Feb. 2020.


\bibitem{li2022generic} J. Li, X. Tang, and C. Hollanti, ``A generic transformation for optimal node repair in MDS array codes over $F_2$," \textit{IEEE Trans. Commun.,} vol. 70, no. 2, pp. 727-738, Feb. 2022.

\bibitem{liu2022generic} Y. Liu, J. Li, and X. Tang, ``A generic transformation to generate MDS array codes with $\delta$-optimal access property," \textit{IEEE Trans. Commun.,} vol. 70, no. 2, pp. 759-768, Feb. 2022.

\bibitem{vajha2023small}
M.~Vajha, S.~Balaji, and P.~V. Kumar, ``Small-d MSR codes with optimal access,
optimal sub-packetization, and linear field size,'' \textit{IEEE Trans. Inform. Theory,} vol. 69, no. 7, pp. 4303-4332, Jul. 2023.

\bibitem{zhang2023vertical} Z. Zhang and L. Zhou, ``A vertical-horizontal framework for building rack-aware regenerating codes," \textit{IEEE Trans. Inform. Theory,} vol. 69, no. 5, pp. 2874--2885, May 2023.

\bibitem{li2023MDS}
J.~Li, Y.~Liu, X.~Tang, Y. S.~Han, B.~Bai, and G.~Zhang, ``M{DS} array codes with (near) optimal repair
bandwidth for all admissible repair degrees,'' \textit{IEEE Trans. Commun.,} vol. 71, no. 10, pp. 5633-5646, Oct. 2023.


\bibitem{wang2023constructing} N. Wang, G. Li, S. Hu, and M. Ye, ``Constructing MSR codes with subpacketization $2^{n/3}$ for $k + 1$ helper nodes,"  \textit{IEEE Trans. Inform. Theory,} vol. 69, no. 6, pp. 3775--3792, June 2023.

\bibitem{li2024msr} G. Li, N. Wang, S. Hu, and M. Ye, ``MSR codes with linear field size and smallest sub-packetization for any number of helper nodes,"  \textit{IEEE Trans. Inform. Theory,} vol. 70, no. 11, pp. 7790--7806, Nov. 2024.



\bibitem{liu2023generic} Y.~Liu, J.~Li, and X.~Tang, ``A generic transformation to enable optimal repair/access MDS array codes with multiple repair degrees," \textit{IEEE Trans. Inform. Theory,} vol. 69, no. 7, pp. 4407-4428, Jul. 2023.


\bibitem{wu2021achievable} T.-Y. Wu, Y. S. Han, Z. Li, B. Bai, G. Zhang, X. Zhang, and X. Wu, ``Achievable lower bound on the optimal access bandwidth of $(k+ 2, k, 2)$-MDS
array code with degraded read friendly," \textit{in 2021 IEEE Information Theory Workshop (ITW)}, pp. 1–5, IEEE, 2021.

\bibitem{ye2020new} M. Ye, ``New constructions of cooperative MSR codes: Reducing node size to exp(O(n))", \textit{IEEE Trans. Inform. Theory,} vol.66, no.12, pp.7457-7464, 2020.

\bibitem{liu2023new} Y. Liu, H. Cai, and
X. Tang, ``A new cooperative repair scheme with $k + 1$ helper nodes for $(n, k)$ Hadamard MSR codes with small sub-packetization", \textit{IEEE Trans. Inform. Theory,} vol. 69, no. 5, pp. 2820-2829, May 2023.

\bibitem{eMSR_d_eq_n-1} J. Li, Y. Liu, and X. Tang, ``A systematic construction of MDS codes with small sub-packetization level and near-optimal repair bandwidth," \textit{IEEE Trans. Inform. Theory,} vol. 67, no. 4, pp. 2162-2180, Apr. 2021.

\bibitem{li2024constructions} L. Li, X. Yu, L. Chen, Y. Dong, and Y. Luo, ``Constructions of binary MDS array codes with optimal repair/access bandwidth", \textit{IEEE Trans. Commun.,} vol.72, no. 6, pp.3113-3125, June 2024.

\bibitem{wang2023rack} J. Wang, D. Zheng, S. Li, and X. Tang, ``Rack-aware MSR codes with error correction capability for multiple erasure tolerance", \textit{IEEE Trans. Inform. Theory,} vol. 69, no. 10, pp. 6428-6442, Oct. 2023. 

\end{thebibliography}
\end{document}